%% file: rule2009-kmt.tex
\documentclass{eptcs}

\usepackage[T1]{fontenc}

\usepackage{amsthm}
\usepackage{amssymb}
\usepackage{xspace}
\usepackage{paralist}
\usepackage{mathpartir} 
\usepackage{proof}

\input{rule2009-kmt-macros.tex}

\newtheorem{thm}{Theorem}[section]
\newtheorem{defn}[thm]{Definition}

\newtheorem{exmp}[thm]{Example}

\title{A Type System for \tom}
\author{Claude Kirchner
\institute{INRIA \\
Centre de Recherche \\INRIA Bordeaux - Sud-Ouest \\
  351, cours de la Lib\'eration, \\
  33405 Talence Cedex France}
\email{Claude.Kirchner@inria.fr}
\and
Pierre-Etienne Moreau
\quad\quad 
Cl\'audia Tavares\thanks{This work was partially supported by The Capes
Foundation, Ministry of Education of Brazil. Cx. postal 365, Bras\'ilia DF
70359-970, Brazil.}
\institute{INRIA \& LORIA\\
615 rue du Jardin Botanique, CS 20101\\
54603 Villers-l\`es-Nancy Cedex France}
\email{Pierre-Etienne.Moreau@loria.fr
\quad\quad
Claudia.Tavares@loria.fr}
}

\def\titlerunning{A Type System for \tom}
\def\authorrunning{Claude Kirchner,  Pierre-Etienne Moreau,  Cl\'audia Tavares}
\begin{document}
\maketitle

\begin{abstract}
  Extending a given language with new dedicated features is a general and quite used
  approach to make the programming language more adapted to problems. Being closer to the
  application, this leads to less programming flaws and easier maintenance.  But of course
  one would still like to perform program analysis on these kinds of extended languages,
  in particular type checking and inference. In this case
  one has to make the typing of the extended features compatible with the ones in the starting language. \\
  The \tom programming language is a typical example of such a situation as it consists of
  an extension of \java that adds pattern matching, more particularly associative pattern
  matching, and reduction strategies. \\
  This paper presents a type system with subtyping for
  \tom, that is compatible with \java's type
  system, and that performs both type checking and type inference.
  We propose an algorithm that checks if all patterns of a \tom program are well-typed. In
  addition, we propose an algorithm based on equality and subtyping constraints that
  infers types of variables occurring in a pattern. Both algorithms are exemplified and
  the proposed type system is showed to be sound and complete.

\end{abstract}

\sloppy


\section{Introduction of the problem: static typing in \tom}
\label{sec:tom}

We consider here the {\tom} language, which is an extension of {\java} that provides rule
based constructs. In particular, any {\java} program is a {\tom} program. We
call this kind of extension {\em formal
islands}~\cite{balland:06,balland:09} where the {\em ocean} consists of {\java} code
and the {\em island} of algebraic patterns. For simplicity, we consider here only two
new {\tom} constructs: a \verb|%match| construct and a
\verb|`| (backquote) construct.

The semantics of \verb|%match| is close to the {\em match} that exists in functional
programming languages, but in an imperative context. A \texttt{\%match} is parameterized
by a list of subjects (\ie\ expressions evaluated to ground terms) and contains a list of
rules. The left-hand side of the rules are patterns built upon constructors and fresh
variables, without any linearity restriction. The right-hand side is {\em not} a
term, but a {\java} statement that is executed when the pattern matches the
subject. However, thanks to the backquote construct (\texttt{`}) a term can be easily built and returned.
In a similar way to the standard \texttt{switch/case} construct, patterns are evaluated from top to
bottom. In contrast to the functional {\em match}, several actions ({\ie} right-hand
sides) may be fired for a given subject as long as no \texttt{return} or \texttt{break} instruction is
executed. To implement a simple reduction step for each rule, it suffices to encode the
left-hand side with a pattern and consider the {\java} statement that returns the right-hand
side.

For example, given the sort \texttt{Nat} and the function symbols \texttt{suc} and
\texttt{zero}, addition and comparison of Peano integers may be encoded as follows: \vs
$$
\begin{array}{c|c}
\begin{minipage}[t]{.43\linewidth}
\begin{footnotesize}
\begin{verbatim}
public Nat plus(Nat t1, Nat t2) {
 %match(t1,t2) {
   x,zero() -> { return `x; }
   x,suc(y) -> { return `suc(plus(x,y)); }
 }
}
\end{verbatim}
\end{footnotesize}
\end{minipage}
\quad&
\begin{minipage}[t]{.55\linewidth}
\begin{footnotesize}
\begin{verbatim}
public boolean greaterThan(Nat t1, Nat t2) {
 %match(t1, t2) {
   x,x           -> { return false; }
   suc(x),zero() -> { return true; }
   zero(),suc(y) -> { return false; }
   suc(x),suc(y) -> { return `greaterThan(x,y); }
 }
}
\end{verbatim}
\end{footnotesize}
\end{minipage}
\end{array}
$$
In this combination of an ocean language (in our case \java) and island features (in our
case abstract data types and matching), it is still an open question to perform type
checking and type inference.

Since we want to allow for type inclusion at the pattern level, the first purpose of this
paper is to present an extension of the signature definition mechanism allowing for
subtypes. In this context we define \textit{\java-like} types and signatures.
Therefore the set of types is the union of \java types and
abstract data types (i.e. \tom types) where multiple inheritance and
overloading are forbidden. For example, given the sorts \texttt{Int}$^+$,
\texttt{Int}$^-$, \texttt{Int} and \texttt{Zero}, the type system accepts the
declaration 
$\texttt{Int}^+ \sub \texttt{Int} ~\wedge~ \texttt{Int}^-  ~\sub \texttt{Int} $
but refuses the declaration 
$\texttt{Zero} \sub \texttt{Int}^+ ~\wedge~ \texttt{Zero} \sub \texttt{Int}^-$.
Moreover, a function symbol \texttt{suc} cannot be overloaded on
both sorts \texttt{Int}$^+$ and \texttt{Int}$^-$. In order to handle those
issues, we propose an algorithm based on unification of equality
constraints \cite{milner:78} and simplification of subtype constraints
\cite{dominic:01,aiken:92,pottier:01}. It infers the types of the variables that occur
in a pattern (\texttt{x} and \texttt{y} in the previous example). Moreover, we
also propose an algorithm
that checks that the patterns occurring in a {\tom} program are correctly typed.

Of course typing systems for algebraic terms and for rewriting has a long history.  It
includes the seminal works done on OBJ, order-sorted algebras~\cite{obj3-88,OBJ2-POPL}
and Maude~\cite{Maude2007}; the works done on feature
algebras~\cite{AitKaciPodelskiSmolkaTCS94} or on membership
constraints~\cite{HintermeierKK-jsc98,DBLP:journals/jsc/Comon98}; and the works on typing
rewriting in higher-order settings like~\cite{BakelF97} or ~\cite{barthe03a}.  Largely
inspired from these works, our contribution here focusses on the appropriate type system
for pattern-matching, possibly modulo associativity, in a \java environment.


\section{Type checking}
\label{sec:checking}

Given a signature~$\FV$, the (simplified) abstract syntax of a {\tom} program is as follows:
$$
\begin{array}{lcl}
        rule    & ::= & cond \longrightarrow action\\
        cond    & ::= & term_1 \match_{[s]} term_2
                   \mid cond_1 \wedge cond_2 \\
        term    & ::= &  x \mid f(term_1,\ldots,term_n)\\
        action  & ::= &  (term_1,\ldots,term_n)
\end{array}
$$
The left-hand side of a rule is a conjunction of matching conditions $term_1 \match_{[s]}
term_2 $ consisting of a pair of terms and where $s$ denotes a sort. We introduce the set
$\caF$ of free symbols. Terms are many-sorted terms composed of variables $x\in\XX$ and
function symbols $f\in\caF$. The set of terms is written $\TFX$. In general, an {\em
  action} is a {\java} statement, but for our purpose it is enough to consider an
abstraction consisting of terms $e_1,\ldots,e_n\in\TFX$
 whose instantiations are
described by the conditions, and used in the {\java} statement.

\begin{exmp}
The last rule of the \verb+greaterThan+ function given above can be represented by the following $rule$ expression:
$$
suc(x)\match_{[\NN]} t_1 \wedge suc(y)\match_{[\NN]} t_2 \longrightarrow (x, y)
$$
\end{exmp}

In a first step, we define $\caS$ as a set of sorts and we consider that a {\em context}~$\Gamma$ is composed of a set of pairs
(variable,sort), and (function symbol,rank):
{\small
$$
\Gamma ::= 
      \varnothing 
      \mid \Gamma_1 \cup \Gamma_2
      \mid x : s 
      \mid f : \sig{s_1, \ldots,s_n}{s}
$$
}
and context access is defined by the function \sortof{\Gamma}{e} $ : \Gamma \times
\TFX \rightarrow \caS$ which returns the sort of term $e$ in the context $\Gamma$:
\begin{center}
{\small
\begin{tabular}{rcl}
        \sortof{\Gamma}{x} 
        $= s$, if $x : s \in \Gamma$ 
        & &
        \sortof{\Gamma}{f(e_1,\ldots,e_n)}      
        $= s$, if $f : \sig{s_1, \ldots, s_n}{s} \in \Gamma$    \\
\end{tabular}}
\end{center}
where $x \in \XX$ and $f \in \caF$.

We denote by $\Gamma(x:s)$
the fact that $x:s$
 belongs to $\Gamma$.
Similarly,
$\Gamma(f:\sig{s_1,\ldots,s_n}{s})$ means that $f:\sig{s_1,\ldots,s_n}{s}$
belongs to $\Gamma$.
In Fig.~\ref{fig:simplechecking} we give a classical type checking system defined by a set of inference rules.
Starting from a context~$\Gamma$ and a rule expression $\pi$, we say that $\pi$
is well-typed if $\pi:\wt$ can be derived by applying the inference rules. $\wt$
is a special sort that corresponds to the well-typedness of a $rule$ or a
condition $cond$.

\begin{figure}[h!]
$$
\begin{array}{|c|}
\hline\\
        \infer[\TVar] 
                {\Gamma(x : s) \vdash x : s}
                {}
        \quad\quad\quad\quad
        \infer[\TFun]
                {\Gamma(f : \sig{s_1,\ldots,s_n}{s}) \vdash f(e_1, \ldots, e_n) : s}
                {\Gamma \vdash e_1 : s_1
                &\quad  
                \ldots
                &\quad
                \Gamma \vdash e_n : s_n}        \\
                                                \\
        \infer[\TMatch]
                {\Gamma \vdash (e_1 ~\match_{[s]}~ e_2) : \wt}
                {\Gamma \vdash e_1 : s
                &\quad  
                \Gamma \vdash e_2 : s}
        \quad\quad\quad\quad
        \infer[\TConj]
                {\Gamma \vdash (cond_1 \wedge \ldots \wedge cond_n) : \wt}
                {\Gamma \vdash (cond_1) : \wt
                &\quad
                \ldots
                &\quad
                \Gamma \vdash (cond_n) : \wt}   \\
                                                \\
        \infer[\TRule]
                {\Gamma \vdash (cond \longrightarrow (e_1, \ldots, e_n)) : \wt}
                {\Gamma \vdash (cond) : \wt
                &\quad
                \Gamma \vdash e_1 : s_1
                &\quad
                \ldots
                &\quad
                \Gamma \vdash e_n : s_n}        \\
                \mbox{if~} \sortof{\Gamma}{e_i} = s_i, \mbox{~for~} i \in [1,n]      \\
                                                \\
\hline
\end{array}
$$
\caption{Simple type checking system.}
\label{fig:simplechecking}
\end{figure}

\subsection{Subtypes and associative-matching}
\label{subsec:subtypes}

In order to introduce subtypes in \tom, we refine $\caS$ as the set of sorts, equipped
with a partial order~$\sub$, called {\em subtyping}. It is a binary relation on $\caS$
that satisfies reflexivity, transitivity and antisymmetry. Moreover, since we allow for
some symbols to be associative, we introduce the set $\caF_v$ of variadic symbols to
denote them.  Now, the set of terms is written $\TFVX$ and terms are many-sorted variadic
terms composed of variables $x\in\XX$ and function symbols $f\in\caF\cup\caF_v$. In the
following, we often write $\ell$ a variadic operator and call it a {\em list}.

We extend matching over lists to be associative. Therefore a pattern matches a subject
considering equality relation modulo flattening. Lists can be denoted by function symbols
$\ell \in \caF_v$ or by variables $x \in \XX$ annotated by $^*$. Such variables, which we
write $x^*$, are called {\em star variables}. So we consider in the following many-sorted
variadic terms composed of variables $x\in\XX$, star variables $x^*$ (where $x \in \XX$)
and function symbols $f\in\caF\cup\caF_v$. Moreover, we define that function symbols $\ell
\in \caF_v$ with variable domain (since they have a variable arity) of sort $s_1$ and
codomain $s$ are written $\ell : \vsig{s_1}{s}$ while star variables $x^*$ are also sorted
and written $x^* : s$.

Since terms built from syntactic and variadic operators can have the same codomain, we
cannot distinguish one from the other only by theirs sorts. However, this is necessary to
know which typing rule applies. Moreover, an insertion of a term can be treated by two
ways: given terms $\ell(e_1),\ell(e_2),\ell_1(e_1) \in \TFVX$ where $\ell, \ell_1 \in
\caF_v$, we have:
\begin{inparaenum}[1)]
\item an insertion of a list $\ell(e_1)$ into a list $\ell(e_2)$ corresponds to a
concatenation of these both lists resulting in $\ell(e_1,e_2)$;
\item an insertion of a list $\ell_1(e_1)$ into a list $\ell(e_2)$ results in $\ell(\ell_1(e_1),e_2)$.
\end{inparaenum}
For that reason, it is important to distinguish the list from the inserted term by its
function symbol in order to define which typing rule concerned for list must be
applied. For this purpose, we introduce a notion of sorts decorated with function symbols,
called {\em types}, to classify terms. The special symbol $?$ is used as decoration when
it is not useful to know what the function symbol is, i.e. when the expected type is known
but not the expected function symbol. This leads to a new set of decorated sorts $\caD$
which is equipped with a partial order $\sub_s$. It is a binary relation on $\caD$ where
$s_1^{g_1} \sub_s s_2^{g_2}$ is equivalent to $s_1 \sub s_2 \wedge (g_1 = g_2 \vee
g_2 = ?)$.

\textit{As pointed out in the introduction, we assume in all that paper that the
  signatures considered do not have multiple inheritance and that we do not allow function
  symbol overloading.} 

Given these notions, we refine the notion of context $\Gamma$ as a set of subtyping
declarations (type,type) and pairs (variable,type), and (function symbol,rank). This is
expressed by the following grammar:
$$
\Gamma ::= 
      \varnothing 
      \mid \Gamma_1 \cup \Gamma_2
      \mid s_1^? \transsub_s s_2^?
      \mid x : s^g 
      \mid x^*: s^{\ell} 
      \mid f : \sig{s_1^?, \ldots,s_n^?}{s^f}
      \mid \ell : \vsig{(s_1^?)}{s^{\ell}}
$$
where $\transsub_s$ corresponds to the reflexive transitive closure of~$\sub_s$ and context access is refined by the function \sortof{\Gamma}{e} $ : \Gamma \times
\TFVX \rightarrow \caD$ which returns the type of term $e$ in the context $\Gamma$:
$$
\begin{array}{rclcrcl}
        \sortof{\Gamma}{x}              
        &=& s^g, \textrm{ if } x : s^g \in \Gamma
        & &
        \sortof{\Gamma}{f(e_1,\ldots,e_n)}      
        &=& s^f, \textrm{ if } f : \sig{s_1^?, \ldots, s_n^?}{s^f} \in \Gamma    \\
        \sortof{\Gamma}{\xx}    
        &=& s^{\ell}, \textrm{ if } x^* : s^{\ell} \in \Gamma
        & &
        \sortof{\Gamma}{\ell(e_1,\ldots,e_n,e)}
        &=& s^{\ell}, \textrm{ if } \ell : \vsig{(s_1^?)}{s^{\ell}} \in \Gamma     \\
\end{array}
$$
where $x \in \XX$, $f \in \caF$, $\ell \in \caF_v$, $g \in \caF\cup\caF_V\cup\{?\}$ and
$s^?,s^f,s^g,s^{\ell} \in \caD$.
 
The context has at most one declaration of type or signature per term since overloading is
forbidden. This means that for $e \in \TFVX$ and $s_1^{g_1}, s_2^{g_2}$ (where
$g_1,g_2 \in \caF\cup\caF_v\cup\{?\}$ and $s_1^{g_1}, s_2^{g_2}
\in \caD$) if $e : s_1^{g_1} \in
\Gamma$ and $e : s_2^{g_2} \in \Gamma$ then $s_1^{g_1} = s_2^{g_2}$. We denote
by $\Gamma(s_1^{g_1} \sub_s s_2^{g_2})$
the fact that $s_1^{g_1} \sub_s s_2^{g_2}$
 belongs to $\Gamma$.

\subsection{Type checking algorithm}
\label{subsec:checking_algo}

In Fig.~\ref{fig:checking} we give a type checking system to many-sorted variadic terms applying
associative matching. The rules are standard except for the use of
decorated types. The most interesting rules are those 
that apply to lists. They are three: [\TEmpty] checks
if a empty list has the same type declared in $\Gamma$; [\TElem] is similar to
[\TFun] but is applied to lists; and [\TMerge] is applied to a concatenation of
two lists of type $s^{\ell}$ in $\Gamma$, resulting in a new list with same type
$s^{\ell}$.
\begin{figure}[h!]
\begin{center}\small
\begin{tabular}{|lll|}
        \hline
&&                                                                                        \\
        \infer[\TVar] 
                {\Gamma(x : s^g) \vdash x : s^g}
                {}
&&
        \infer[\TSVar] 
                {\Gamma(\xx : s^{\ell}) \vdash \xx : s^{\ell}}
                {}                                                                      \\
        where $g \in \caF\cup\caF_v \cup \{?\}$
&&                                      \\
&&                                                                                        \\
       \infer[\TFun]
                {\Gamma(f : \sig{s_1^?,\ldots,s_n^?}{s^f}) \vdash f(e_1, \ldots, e_n) : s^f}
                {\Gamma \vdash e_1 : s_1^?
                &\quad  
                \ldots
                &\quad
                \Gamma \vdash e_n : s_n^?}  
&&                                      
       \infer[\TEmpty]
                {\Gamma(\ell : \vsig{(s_1^?)}{s^{\ell}}) \vdash \ell() : s^{\ell}}
                {} \\
&&                                                                                        \\
       \infer[\TElem]
                {\Gamma(\ell : \vsig{(s_1^?)}{s^{\ell}}) \vdash \ell(e_1, \dots, e_n,
e) : s^{\ell}}
                {\Gamma \vdash \ell(e_1, \ldots, e_n) : s^{\ell}
                &\quad  
                \Gamma \vdash e : s_1^?}  
&&
       \infer[\TMerge]
                {\Gamma(\ell : \vsig{(s_1^?)}{s^{\ell}}) \vdash \ell(e_1, \dots, e_n, e) : s^{\ell}}
                {\Gamma \vdash \ell(e_1, \ldots, e_n) : s^{\ell}
                &\quad
                \Gamma \vdash e : s^{\ell}} \\
if \sortof{\Gamma}{e} $\neq s^{\ell}$ and $e \neq \xx$
&&
if \sortof{\Gamma}{e} $= s^{\ell}$          \\
&&                                                                                        \\
        \infer[\Sub]
                {\Gamma(s_1^{g_1} \sub_s s^g) \vdash e : s^g}
                {\Gamma \vdash e : s_1^{g_1}}
&&
        \infer[\Gen]
                {\Gamma \vdash e : s^?}
                {\Gamma \vdash e : s^h}                                                 \\
        where $g, g_1 \in \caF\cup\caF_v\cup\{?\}$
&&
        if \sortof{\Gamma}{e} $= s^h$, where $h \in \caF\cup\caF_v$                       \\
&&                                                                                        \\
       \infer[\TMatch]
                {\Gamma \vdash (e_1 ~\match_{[s^?]}~ e_2) : \wt}
                {\Gamma \vdash e_1 : s^?
                &\quad  
                \Gamma \vdash e_2 : s^?}
&&                                                                                       
       \infer[\TConj]
                {\Gamma \vdash (cond_1 \wedge \ldots \wedge cond_n) : \wt}
                {\Gamma \vdash (cond_1) : \wt
                &\quad
                \ldots
                &\quad
                \Gamma \vdash (cond_n) : \wt}                                           \\
&&                                                                                        \\
\multicolumn{3}{|c|}{
        \infer[\TRule]
                {\Gamma \vdash (cond \longrightarrow (e_1, \ldots, e_n)) : \wt}
                {\Gamma \vdash (cond) : \wt
                &\quad
                \Gamma \vdash e_1 : s_1^{g_1}
                &\quad
                \ldots
                &\quad
                \Gamma \vdash e_n : s_n^{g_n}}}                                  \\
\multicolumn{3}{|c|}{
                if $\sortof{\Gamma}{e_i} = s_i^{g_i}$, where $g_i \in
\caF\cup\caF_v\cup\{?\}$ for $i \in [1,n]$}      \\
&&                                                                                        \\
\hline
\end{tabular}
\end{center}
\caption{Type checking rules.}
\label{fig:checking}
\end{figure}

The type checking algorithm reads derivations bottom-up. Since the rule [\Sub] can be applied
to any kind of term, we consider a strategy where it is applied iff no other
typing rule can be applied. In practice, [\Sub] will be
combined with [\TVar], [\TFun] and [\TElem] and
the type $s_1^?$ which appears in the premise will be defined
according to the result of function \sortof{\Gamma}{e}. The algorithm stops if
it reaches the [\TVar] or [\TSVar] cases, ensuring that the original expression is well-typed, or if
none of the type checking rules can be applied, raising an error.

\begin{exmp}
\label{exmp:checking}
Let $\Gamma = \{\ell : \vsig{(\ZZ^?)}{\ZZ^{\ell}}, one : \sig{}{\NN^{one}}, x^* :
\ZZ^{\ell}, z^* : \ZZ^{\ell}, y :
\ZZ^?, \NN^? \sub_s \ZZ^?\}$. Then the expression $\ell(x^*,y,z^*)
\match_{[\ZZ^?]} \ell(one())
\longrightarrow (y)$ is well-typed
and its deduction tree is given in Fig.~\ref{fig:example-tc}

\begin{figure}
\begin{center}
{\scriptsize
\begin{mathpar}
        \inferrule*[Right={\scriptsize T-Rule},rightskip=51mm]
                {\inferrule*[leftskip=13mm,Right={\scriptsize T-Match}]
                        {\inferrule*[vdots=25mm,Right={\scriptsize
        T-Gen},leftskip=22mm,rightskip=30mm]
                                {\inferrule*[Right={\scriptsize
        T-Merge}]
                                        {\inferrule*[leftskip=19mm,Right={\scriptsize T-Elem}]
                                                {\inferrule*[leftskip=13mm,Right={\scriptsize T-Merge}]
                                                        {\inferrule*[Right={\scriptsize T-Empty}]
                                                                {~}
                                                                {\Gamma \vdash
\ell() : \ZZ^{\ell}}
                                                        \\ \quad \\
                                                        \inferrule*[Right={\scriptsize T-SVar}]
                                                                {~}
                                                                {\Gamma
                        \vdash x^* : \ZZ^{\ell}}
                                                        }
                                                        {\Gamma \vdash \ell(x^*) :
\ZZ^{\ell}}
                                                \\ \quad \\
                                                \inferrule*[Right={\scriptsize T-Var}]
                                                        {~}
                                                        {\Gamma \vdash y : \ZZ^?}
                                                }
                                                {\Gamma \vdash \ell(x^*,y) :
\ZZ^{\ell}}
                                        \\ \quad \\
                                        \inferrule*[Right={\scriptsize T-SVar}]
                                                {~}
                                                {\Gamma \vdash z^* : \ZZ^{\ell}}
                                        }
                                        {\Gamma \vdash \ell(x^*,y,z^*) : \ZZ^{\ell}}
                                }
                                {\Gamma \vdash \ell(x^*,y,z^*) : \ZZ^?}
                        \\
                        \inferrule*[Right={\scriptsize T-Gen}]
                                {\inferrule*[Right={\scriptsize T-Elem}]
                                        {\inferrule*[Right={\scriptsize T-Empty}]
                                                {~}
                                                {\Gamma \vdash \ell() : \ZZ^{\ell}}
                                        \\ \quad \\
                                        \inferrule*[Right={\scriptsize Sub}]
                                                {\inferrule*[Right={\scriptsize Gen}]
                                                        {\inferrule*[Right={\scriptsize T-Fun}]
                                                                {~}
                                                                {\Gamma \vdash one() :
                        \NN^{one}}
                                                        }
                                                        {\Gamma \vdash one() :
                        \NN^?}
                                                }
                                                {\Gamma \vdash one() :
                        \ZZ^?}
                                        }
                                        {\Gamma \vdash \ell(one()) : \ZZ^{\ell}}
                                }
                                {\Gamma \vdash \ell(one()) : \ZZ^?}
                        }
                        {\Gamma \vdash (\ell(x^*,y,z^*) \match_{[\ZZ^?]} \ell(one())) : \wt}
                \\ \quad \\
                \inferrule*[Right={\scriptsize T-Var}]
                        {~}
                        {\Gamma \vdash y : \ZZ^?}
                }
                {\Gamma \vdash (\ell(x^*,y,z^*) \match_{[\ZZ^?]} \ell(one())
\longrightarrow (y)) : \wt}
\end{mathpar}}
\end{center}
\caption{Type checking example.}
\label{fig:example-tc}
\end{figure}
\end{exmp}


\section{Type inference}
\label{sec:inference}

The type system presented in Section~\ref{sec:checking} needs rules to
control its use in order to find the expected deduction tree of an expression. Without these rules it is possible to find more than one
deduction tree for the same expression. For instance, in Example
\ref{exmp:checking}, the rule [\Sub] can be applied to the leaves resulting of application of
rule [\TVar]. The resulting tree will still be a valid deduction tree since the
variables in the leaves will have type $\NN^?$ instead of type $\ZZ^?$
declared in the context and $\NN^? \sub_s \ZZ^?$. For that reason, we are interested in defining another type system
able to infer the most general types of terms. We add type variables in the set
of types (defined up to here as a set of decorated sorts) to describe a
possibly infinite set of decorated sorts. The set of types
$\TYDV$ is given by a set of decorated sorts $\caD$, a set of type variables
$\VV$ and a special sort $\wt$:
$$
\tau ::= \alpha \mid s^g \mid \wt
$$
where $\tau \in \TYDV$, $\alpha \in \VV$, $g \in \caF\cup\caF_v\cup\{?\}$ and $s^g \in \caD$. 

In order to build the subtyping rule into the rules, we use a
{\em constraint set} \CC to store all equality and subtyping constraints. These
constraints limit types that terms can have. The language $\caC$ is built from
the set of types $\TYDV$ and the operators ``$=_s$'' (equality) and
``$\sub_s$'' (extension to $\TYDV$ of the partial order defined in Subsection~\ref{subsec:subtypes}):
$$
c ::= \tau_1 =_s \tau_2 \mid \tau_1 \sub_s \tau_2
$$
where $c \in \caC$, $\tau_1,\tau_2 \in \TYDV$.

A substitution $\sigma$ is said to {\em satisfy} an equation $\tau_1 =_s \tau_2$
if $\sigma\tau_1 = \sigma\tau_2$. Moreover, $\sigma$ is said
to {\em satisfy} a subtype relation $\tau_1 \sub_s \tau_2$ if $\sigma\tau_1
\sub_s \sigma\tau_2$.

Thus, $\sigma$ is a {\em solution} for $\CC$ if
it satisfies all constraints in $\CC$. This is written $\sigma \models \CC$. The set $\VC$ denotes the set of type
variables in $\CC$. 

Constraints are calculated according to the application of rules of
 type inference system (see Fig.~\ref{fig:inference}) where we
can read the judgment $\Gamma \ctvdash e : \tau \bullet \CC$ as ``the term $e$ has type $\tau$
under assumptions $\Gamma$ whenever the constraints \CC are
satisfied''. More formally, this judgment states that $
\forall \sigma \centerdot (\sigma \models \CC \rightarrow \sigma\Gamma \vdash e : \sigma\tau)$.

\subsection{Type inference algorithm}
\label{subsec:inference_algo}

In Fig. \ref{fig:inference} we give a type inference system with constraints.
In order to infer the type of a given expression~$\pi$, the context $\Gamma$ is initialized to: 
\begin{inparaenum}[1)]
\item subtype declarations of the form $s_1^? \sub_s s_2^?$ where $s_1^?$ and $s_2^? \in \caD$;
\item a pair of the form $(f : \sig{s_1^?,\ldots,s_n^?}{s^f})$ for each
syntactic operator $f$ occurring in $\pi$ where $s_i^?,s^f \in \caD$ for $i \in [1,n]$;
\item a pair of the form $(\ell : \vsig{s_1^?}{s^{\ell}})$ for each variadic
operator $\ell$ occurring in $\pi$ where $s_1^?,s^{\ell} \in \caD$;
\item a pair of the form $(x : \alpha)$ for each variable $x$ occurring in $\pi$ where $\alpha \in \VV$ is a fresh type variable;
\item a pair of the form $(\xx : \alpha)$ for each star variable $x^*$ occurring
in $\pi$ where $\alpha \in \VV$ is a fresh type variable.
\end{inparaenum}
Moreover, each type variable introduced in a sub-derivation is a fresh type variable and
the fresh type variables in different sub-derivations are distinct. As in Section~\ref{subsec:checking_algo}, we explain the rules concerning lists: [\CTEmpty] infers for an empty list $\ell()$ a type variable
$\alpha$ with the constraint $\alpha = s^{\ell}$, $s^{\ell}$ given by the rank of
$\ell$; [\CTElem] treats applications of lists to elements which are neither lists
with the same function symbol nor star variables; [\CTMerge] is applied to
concatenate two lists of same type $s^{\ell}$; and [\CTStar] is applied to
concatenate a list and a star variable of the same type $s^{\ell}$.

\begin{figure}[h!]
\begin{center}

\begin{tabular}{|c|}
        \hline
                                                                                        \\
                                                                                        \\
        \infer[\CTVar] 
                {\Gamma(x : \tau) \ctvdash  x : \alpha 
                        \bullet \{\alpha =_s \tau\}}
                {~}
                \quad\quad\quad\quad
        \infer[\CTSVar] 
                {\Gamma(\xx : \alpha_1) \ctvdash  \xx : \alpha 
                        \bullet \{\alpha_1 =_s \alpha\}}
                {~}             \\
                                \\
        \infer[\CTFun]
                {\Gamma(f : \sig{s_1^?,\ldots,s_n^?}{s^f}) 
                        \ctvdash f(e_1,\ldots,e_n) : \alpha 
                        \bullet  \{\alpha =_s s^f\} 
                        \bigcup \limits_{i = 1}^n \CC_i 
                        \cup \{\alpha_i \sub_s s_i^?\}}   
                {\Gamma \ctvdash e_1 : \alpha_1 \bullet \CC_1 
                &
                \ldots
                &
                \Gamma \ctvdash e_n : \alpha_n \bullet \CC_n}           \\
                                                                        \\
                                                                        \\
        \infer[\CTEmpty]
                {\Gamma(\ell : \vsig{(s_1^?)}{s^{\ell}}) \ctvdash \ell(): \alpha 
                        \bullet  \{\alpha =_s s^{\ell}\}}
                {~}                                                      \\
                                                                         \\
        \infer[\CTElem]
                {\Gamma(\ell : \vsig{(s_1^?)}{s^{\ell}}) 
                        \ctvdash \ell(e_1,\ldots,e_n,e) : \alpha 
                        \bullet  \{\alpha =_s s^{\ell}, \alpha_1 \sub_s s_1^?\}
                        \cup \CC_1 \cup \CC_2} 
                {\Gamma \ctvdash \ell(e_1,\ldots,e_n) : \alpha \bullet \CC_1
                &\quad
                \Gamma \ctvdash e : \alpha_1 \bullet \CC_2}             \\
         if $\sortof{\Gamma}{e} \neq s^{\ell}$ and $e \neq \xx$         \\
                                                                        \\
        \infer[\CTMerge]
                {\Gamma(\ell : \vsig{(s_1^?)}{s^{\ell}}) 
                        \ctvdash \ell(e_1,\ldots,e_n, e) : \alpha 
                        \bullet  \{\alpha =_s s^{\ell}\} \cup \CC_1 \cup \CC_2}
                {\Gamma \ctvdash \ell(e_1,\ldots,e_n) : \alpha \bullet \CC_1
                &\quad
                \Gamma \ctvdash e : \alpha \bullet \CC_2}               \\
         if $\sortof{\Gamma}{e} = s^{\ell}$                             \\
                                                                        \\
        \infer[\CTStar]
                {\Gamma(\ell : \vsig{(s_1^?)}{s^{\ell}}) 
                        \ctvdash \ell(e_1,\ldots,e_n, \xx) : \alpha 
                        \bullet  \{\alpha =_s s^{\ell}\} \cup \CC_1 \cup \CC_2}
                {\Gamma \ctvdash \ell(e_1,\ldots,e_n) : \alpha \bullet \CC_1
                &\quad  
                \Gamma \ctvdash \xx : \alpha \bullet \CC_2}             \\
                                                                        \\
        \infer[\CTMatch]
                {\Gamma \ctvdash (e_1 ~\match_{[\tau]}~ e_2) : \wt 
                        \bullet  \{\alpha_1 \sub_s \tau, \alpha_2 =_s \tau\} 
                        \cup \CC_1 \cup \CC_2}
                {\Gamma \ctvdash e_1 : \alpha_1 \bullet \CC_1
                &\quad
                \Gamma \ctvdash e_2 : \alpha_2 \bullet \CC_2}           \\
                                                                        \\
        \infer[\CTConj]
                {\Gamma \ctvdash (cond_1 \wedge \ldots \wedge cond_n) : \wt 
                        \bullet  \bigcup \limits_{i = 1}^n \CC_i}
                {\Gamma \ctvdash (cond_1) : \wt \bullet \CC_1
                &\quad
                \ldots
                &\quad
                \Gamma \ctvdash (cond_n) : \wt \bullet \CC_n}           \\
                                                                        \\
        \infer[\CTRule]
                {\Gamma \ctvdash (cond \longrightarrow (e_1,\ldots,e_n)) : \wt \bullet  \CC_{cond} \bigcup \limits_{i=1}^n C_i }
                {\Gamma \ctvdash (cond) : \wt \bullet \CC_{cond}
                &\quad
                \Gamma \ctvdash e_1 : \tau_1 \bullet \CC_{1}
                &
                \ldots
                &
                \Gamma \ctvdash e_n : \tau_n \bullet \CC_{n}}      \\
                if $\sortof{\Gamma}{e_i} = \tau_i$, for $i \in [1,n]$
                where $\tau_i \in \TYDV$   \\              
                                                                        \\
        \hline                                  
\end{tabular}
\end{center}
\caption{Type inference rules.}
\label{fig:inference}
\end{figure}

\begin{exmp}
\label{exmp:inference}
Let $\Gamma = \{\ell : \vsig{(\ZZ^?)}{\ZZ^{\ell}}, one :
\sig{}{\NN^{one}}, x^* : \alpha_1, y : \alpha_2, z^* : \alpha_3, \NN^? \sub_s
\ZZ^?\}$. Then the expression
$\ell(x^*,y,z^*) \match_{[\alpha_4]} \ell(one()) \longrightarrow (y)$ is well-typed
and the deduction tree is given in Fig.~\ref{fig:example-ti}.

\begin{figure}
\begin{center}
{\scriptsize
\begin{mathpar}
        \inferrule*[width=5cm,Right={\scriptsize CT-Star},rightskip=5mm]
                {\inferrule*[width=10cm,vdots=10mm,leftskip=35mm,Right={\scriptsize
CT-Elem},rightskip=70mm]
                        {\inferrule*[width=10cm,leftskip=5mm,Right={\scriptsize CT-Star}]
                                {\inferrule*[Right={\scriptsize CT-Empty}]
                                        {~}
                                        {\Gamma \ctvdash \ell() : \alpha_5 
                                                \bullet \CC_3 = \{\alpha_5 =_s \ZZ^{\ell}\}}
                                \\ \quad \quad \\
                                \inferrule*[Right={\scriptsize CT-SVar}]
                                        {~}
                                        {\Gamma \ctvdash x^* : \alpha_5 
                                                \bullet \CC_4 = \{\alpha_5 =_s \alpha_1\}}}
                                {\Gamma \ctvdash \ell(x^*) : \alpha_5 
                                        \bullet C_2 = \{\alpha_5 =_s \ZZ^{\ell}\}
                                        \cup \CC_3 \cup \CC_4}
                        ~~~~~~~~~~~~~~~~~~~~
                        \inferrule*[Right={\scriptsize CT-Var}]
                                {~}
                                {\Gamma \ctvdash y : \alpha_8 
                                        \bullet \CC_5 = \{\alpha_8 =_s \alpha_2\}}}
                        {\Gamma \ctvdash \ell(x^*,y) : \alpha_5 
                                \bullet C_1 = \{\alpha_5 =_s \ZZ^{\ell}, 
                                \alpha_8 \sub_s \ZZ^?\} \cup C_2 \cup \CC_5}
                ~~~~~~~~~~~~~~~~~
                \inferrule*[Right={\scriptsize CT-SVar}]
                        {~}
                        {\Gamma \ctvdash z^* : \alpha_5 
                                \bullet \CC_6 = \{\alpha_5 =_s \alpha_3\}}}
                {\Gamma \ctvdash \ell(x^*,y,z^*) : \alpha_5 
                        \bullet C_p = \{\alpha_5 =_s \ZZ^{\ell}\} \cup C_1 \cup \CC_6}
\end{mathpar}
\textrm{\textbf{(1)}} \\

\begin{mathpar}
        \inferrule*[width=10cm,Right={\scriptsize CT-Elem}]
                {\inferrule*[Right={\scriptsize CT-Empty}]
                        {~}
                        {\Gamma \ctvdash \ell() : \alpha_6 
                                \bullet \CC_7 =\{\alpha_6 =_s \ZZ^{\ell}\}}
                \\ \quad \quad \\
                \inferrule*[Right={\scriptsize CT-Fun}]
                        {~}
                        {\Gamma \ctvdash one() : \alpha_7 
                                \bullet \CC_8\{\alpha_7 =_s \NN^{one}\}}}
                {\Gamma \ctvdash \ell(one()) : \alpha_6 
                        \bullet C_s = \{\alpha_6 =_s \ZZ^{\ell}, \alpha_7 \sub_s
\ZZ^?\} \cup \CC_7 \cup \CC_8}
\end{mathpar}
\textrm{\textbf{(2)}} \\

\begin{mathpar}
        \inferrule*[width=19cm,Right={\scriptsize CT-Rule}]
                {\inferrule*[Right={\scriptsize CT-Match}]
                        {\textbf{(1)~~~~~~~~} 
                        \\ 
                        ~~~~~~~~~~\textbf{(2)}}
                        {\Gamma \ctvdash (\ell(x^*,y,z^*) 
                                \match_{[\alpha_4]} \ell(one())) : \wt 
                                \bullet \CC_{cond} =\{\alpha_5 \sub_s \alpha_4,
                                                \alpha_6 =_s \alpha_4\}}
		\\ \quad \quad \quad \quad
		\inferrule*[Right={\scriptsize CT-Var}]
			{~}
			{\Gamma \ctvdash y : \alpha_9 
                                        \bullet \CC_{10} = \{\alpha_9 =_s \alpha_2\}}
		}
                {\Gamma \ctvdash (\ell(x^*,y,z^*) 
                        \match_{[\alpha_4]} \ell(one()) \longrightarrow (y)) : \wt 
                        \bullet \CC_r = \{\alpha_2 =_s \alpha_2\} \cup \CC_{cond}
\cup \CC_{10}}
\end{mathpar}
}
\end{center}
\caption{Type inference example.}
\label{fig:example-ti}
\end{figure}
\end{exmp}


\subsection{Constraint resolution}
\label{subsec: resol}

In Fig.~\ref{fig:fail_rules} we propose an algorithm to decide whether a given
constraint set $\CC$ has a solution, where $g_1, g_2 \in \caF\cup\caF_v\cup\{?\}$. 
We denote by $s^{g_1} \sub_s s'^{g_2} \in \Gamma$ the fact that there exists $s_1,\ldots,s_n$ such that 
$s^?\sub s_1^?\in\Gamma$, 
$s_1^?\sub s_2^?\in\Gamma$, \ldots,
$s_n^?\sub s'^?\in\Gamma$ and $(g_1=g_2 \textrm{ or } g_2=?)$.
If the algorithm stops without failure then $\CC$ is said to
be in {\em solved form}. 
\begin{figure}[h!]
\begin{center}

\begin{tabular}{|clcl|}
        \hline
                &       &       &       \\
        (1)     & $\{s_1^{g_1} \sub_s \alpha , \alpha \sub_s s_2^{g_2}\} \uplus \CC'$
                & $\Longrightarrow$     
                & $fail$ if $s_1^{g_1} \sub_s s_2^{g_2} \notin \Gamma$  \\
        (2)     & $\{s_1^{g_1} \sub_s \alpha , s_2^{g_2} \sub_s \alpha\} \uplus \CC'$
                & $\Longrightarrow$     
                & $fail$ if $\not\exists s 
                                \centerdot (s_1^{g_1} \sub_s s^? \in \Gamma 
                                \wedge s_2^{g_2} \sub_s s^? \in \Gamma)$  \\
        (3)     & $\{\alpha \sub_s s_1^{g_1}, \alpha \sub_s s_2^{g_2}\} \uplus \CC'$
                & $\Longrightarrow$     
                & $fail$ if $(s_1^{g_1} \sub_s s_2^{g_2} \notin \Gamma
                         \wedge s_2^{g_2} \sub_s s_1^{g_1} \notin \Gamma)$   \\
        (4)     & $\{s_1^{g_1} \sub_s s_2^{g_2}\} \uplus \CC'$  
                & $\Longrightarrow$     
                & $fail$ if $s_1^{g_1} \sub_s s_2^{g_2} \notin \Gamma$ \\
        (5)     & $\{s_1^{g_1} = s_2^{g_2}\} \uplus \CC'$  
                & $\Longrightarrow$     
                & $fail$ if $s_1 \neq s_2 \vee g_1\neq g_2$ \\
                &       &       &       \\
        \hline
\end{tabular}
\end{center}
\caption{Rules for detection of errors in a constraint set $\CC$.}
\label{fig:fail_rules}
\end{figure}

While solving a constraint set $\CC$ we wish to make sure, after each
application of a constraint resolution rule, that the constraint set at hand is
satisfiable, so as to detect
errors as soon as possible. Therefore we must combine the rules for error detection and constraint
resolution in order to keep $\CC$ in solved form. The rules for the constraint resolution algorithm are provided in Fig.
\ref{fig:constraint_rules}, where $g, g_1, g_2 \in \caF\cup\caF_v\cup\{?\}$. The
rules (1)-(14) are recursively applied over $\CC$. More precisely, rules (1)-(3) work as a garbage
collector removing constraints that are no more useful. Rules (4) and (5)
generate $\sigma$. Rules (6) and (7) generate more simplified constraints. Rules (8)-(12)
generate $\sigma$ and simplified constraints by antisymmetric and transitive
subtype closure. Rules (13) and (14) are applied when none of previous rules can be applied
generating a new $\sigma$ from a constraint over a type variable that has no
other constraints. The algorithm stops if: a rule returns $\CC =
\varnothing$, then the algorithm returns 
the solution $\sigma$; if $\CC$ reaches a non-solved form, then the algorithm for detection of
errors returns $fail$; or if $\CC$ reaches a normal form different from the empty set,
then the algorithm returns an error. We say that the algorithm is \textit{failing} if it
returns either \textit{fails} or an error.

\begin{figure}[h]
\begin{center}
\begin{tabular}{|clcl|}
        \hline
                &       &       &                                       \\
        (1)     & $\{\tau =_s \tau \} \uplus \CC', \sigma$                      
                & $\Longrightarrow$     
                & $\CC', \sigma$                                                \\
        (2)     & $\{\tau \sub_s \tau \} \uplus \CC', \sigma$   
                & $\Longrightarrow$     
                & $\CC', \sigma$                                                \\
        (3)     & $\{s_1^{g_1} \sub_s s_2^{g_2} \} \uplus \CC', \sigma$ 
                & $\Longrightarrow$     
                & $\CC', \sigma$ if $s_1^{g_1} \sub_s s_2^{g_2} \in \Gamma$     \\
                &       &       &                                       \\
        (4)     & $\{\alpha =_s \tau\} \uplus \CC', \sigma$
                & $\Longrightarrow$     
                & $[\alpha \mapsto \tau]\CC', \{\alpha \mapsto \tau\} \cup \sigma$                           \\
        (5)     & $\{\tau =_s \alpha\} \uplus \CC', \sigma$     
                & $\Longrightarrow$     
                & $[\alpha \mapsto \tau]\CC', \{\alpha \mapsto \tau\} \cup \sigma$                           \\
                &       &       &                                       \\
        (6)    & $\{s_1^{g_1} \sub_s \alpha , s_2^{g_2} \sub_s \alpha\} 
                        \uplus \CC', \sigma$
                & $\Longrightarrow$     
                & $\{s^? \sub_s \alpha\} \cup \CC', \sigma$
                        if $\exists s 
                                \centerdot (s_1^{g_1} \sub_s s^? \in \Gamma 
                                \wedge s_2^{g_2} \sub_s s^? \in \Gamma)$\\
        (7a)   & $\{\alpha \sub_s s_1^{g_1}, \alpha \sub_s s_2^{g_2}\} 
                        \uplus \CC', \sigma$
                & $\Longrightarrow$     
                & $\{\alpha \sub_s s_1^{g_1}\} \cup \CC', \sigma$
                        if $(s_1^{g_1} \sub_s s_2^{g_2} \in \Gamma)$\\ 
        (7b)   & $\{\alpha \sub_s s_1^{g_1}, \alpha \sub_s s_2^{g_2}\} 
                        \uplus \CC', \sigma$
                & $\Longrightarrow$     
                & $\{\alpha \sub_s s_2^{g_2}\} \cup \CC', \sigma$
                        if $(s_2^{g_2} \sub_s s_1^{g_1} \in \Gamma)$\\ 
                &       &       &                                       \\   
        (8)     & $\{\tau_1 \sub_s \tau_2, \tau_2 \sub_s \tau_1\} 
                        \uplus \CC', \sigma$   
                & $\Longrightarrow$     
                & $\{\tau_1 =_s \tau_2\} \cup \CC', \sigma$                     \\
        (9)     & $\{\alpha_1 \sub_s \alpha , \alpha \sub_s \alpha_2\} 
                        \uplus \CC', \sigma$
                & $\Longrightarrow$     
                & $\{\alpha_1 \sub_s \alpha_2\} 
                        \cup [\alpha \mapsto \alpha_2]\CC', \{\alpha \mapsto
\alpha_2\} \cup \sigma$             \\
        (10)     & $\{s^g \sub_s \alpha , \alpha \sub_s \alpha_1\}
                        \uplus \CC', \sigma$
                & $\Longrightarrow$     
                & $\{s^g \sub_s \alpha_1\} 
                        \cup [\alpha \mapsto \alpha_1]\CC', \{\alpha \mapsto
\alpha_1\} \cup \sigma$             \\
        (11)     & $\{\alpha_1 \sub_s \alpha , \alpha \sub_s s^g\} 
                        \uplus \CC', \sigma$
                & $\Longrightarrow$     
                & $\{\alpha_1 \sub_s s^g\}
                        \cup [\alpha \mapsto \alpha_1]\CC', \{\alpha \mapsto
\alpha_1\} \cup \sigma$             \\
        (12)    & $\{s_1^{g_1} \sub_s \alpha , \alpha \sub_s s_2^{g_2}\} 
                        \uplus \CC', \sigma$
                & $\Longrightarrow$     
                & $[\alpha \mapsto s_2^{g_2}]\CC', \{\alpha \mapsto
s_2^{g_2}\} \cup \sigma$
                        if $s_1^{g_1} \sub_s s_2^{g_2} \in \Gamma$      \\
                &       &       &                                       \\
   
        (13)    & $\{\alpha \sub_s \tau\} \uplus \CC', \sigma$  
                & $\Longrightarrow$     
                & $\CC', \{\alpha \mapsto
\tau\} \cup \sigma$ if $\alpha \notin \VV(\CC')$                   \\
        (14)    & $\{\tau \sub_s \alpha\} \uplus \CC', \sigma$ 
                & $\Longrightarrow$     
                & $\CC', \{\alpha \mapsto
\tau\} \cup \sigma$ if $\alpha \notin \VV(\CC')$                   \\
                &       &       &                                       \\
        \hline
\end{tabular}
\end{center}
\caption{Constraint resolution rules in context $\Gamma$.}
\label{fig:constraint_rules}
\end{figure}
\begin{exmp}
\label{exmp:resolution}

Let $\Gamma = \{\ell : \vsig{(\ZZ^?)}{\ZZ^{\ell}}, one :
\sig{}{\NN^{one}}, x^* : \alpha_1, y : \alpha_2, z^* : \alpha_3, \NN^? \sub_s
\ZZ^?\}$ and
$\CC_{cond} = \{\alpha_5 =_s
\ZZ^{\ell}, \alpha_{10} =_s \alpha_1, \alpha_5 =_s \ZZ^{\ell}, \alpha_{10} =_s
\ZZ^{\ell}, \alpha_9 =_s
\alpha_2, \alpha_5 =_s \ZZ^{\ell},
\alpha_9 \sub_s \ZZ^?, \alpha_8 =_s \alpha_3, \alpha_5 =_s \ZZ^{\ell}, \alpha_8
=_s
\ZZ^{\ell},
\alpha_6 =_s \ZZ^{\ell},
\alpha_7 =_s \NN^{one}, \alpha_6 =_s \ZZ^{\ell}, \alpha_7 \sub_s \ZZ^?, \alpha_5 \sub_s
\alpha_4, \alpha_6 =_s \alpha_4, \alpha_2 =_s \alpha_2\}$ from the Example
\ref{exmp:inference}. Let $\sigma = \varnothing$ and $\CC = \CC_{cond}$. The constraint resolution algorithm
starts by:  
\begin{enumerate}
\item Application of sequence of rules (4), (1) and (5) generating 
$\{\alpha_2 \sub_s \ZZ^?, \NN^{one} \sub_s \ZZ^?, \ZZ^{\ell} \sub_s
\ZZ^{\ell}\} \cup \CC$
and
$\{\alpha_5 \mapsto \ZZ^{\ell}, \alpha_{10} \mapsto \alpha_1,
\alpha_1 \mapsto \ZZ^{\ell}, \alpha_9 \mapsto \alpha_2, \alpha_8 \mapsto \alpha_3,
\alpha_3 \mapsto \ZZ^{\ell}, \alpha_6
\mapsto \ZZ^{\ell}, \alpha_7 \mapsto \NN^{one}, \alpha_4 \mapsto \ZZ^{\ell}\} \cup \sigma$

\item Application of rules (1), (2) and (3) generating $\{\alpha_2 \sub_s \ZZ^?\}$ and $\sigma$;

\item Application of rule (13) generating $\varnothing$ and $\{\alpha_2 \mapsto
\ZZ^?\} \cup \sigma$, the algorithm then stops
and returns $\sigma$ providing a substitution for all type variables in
the deduction tree of $\ell(x^*,y,z^*) \match_{[\alpha_4]} \ell(one())
\longrightarrow (y)$.
\end{enumerate}
\end{exmp}


\section{Properties}
\label{sec:properties}

Since our type checking system and our type inference system address the same
issue, we must
check two properties. First, we show that every typing judgment that can be
derived from the inference rules also follows from the checking rules (Theorem
\ref{thm:soundness}), in particular the soundness. Then we show that a solution
given by the checking rules
can be extended to a solution proposed by the inference rules (Theorem
\ref{thm:completeness}).

\begin{defn}[Solution]
Let $\Gamma$ be a context and $e$ a term. 
\begin{itemize}
\item A {\em solution} for $(\Gamma,e)$ is a pair $(\sigma, T_1)$ such that $\sigma\Gamma
  \vdash \sigma e : T_1$, where $T_1 \in \caD \cup \{\wt\}$.
\item Assuming a well-formed sequent $\Gamma \vdash e : \tau \bullet \CC$, a
{\em solution} for
  $(\Gamma,e,\tau,\CC)$ is a pair $(\sigma,T_2)$ such that $\sigma$ satisfies $\CC$ and
  $\sigma\tau \sub_s T_2$, where $T_2 \in \caD\cup\{\wt\}$ and $\tau \in \TYDV$.
\end{itemize}
\end{defn}

\begin{thm}[Soundness of constraint typing]
\label{thm:soundness}
Suppose that $\Gamma \ctvdash e : \tau \bullet \CC$ is a valid sequent. If $(\sigma, s^g)$ is a
solution for $(\Gamma, e, \tau, \CC)$, then it is also a solution for
$(\Gamma,e)$ (i.e. $e$ is well-typed in $\Gamma$).
\end{thm}

\begin{proof}
By induction on the given constraint typing derivation for $\Gamma
\ctvdash e : \tau \bullet \CC$. We just detail the most noteworthy cases of this
proof.

\begin{flushleft}
\begin{tabular}{lll}
        \textit{Case} $\CTElem$:        & $e = \ell(a_1,\ldots,a_n,a)$
                                & $\tau = \alpha$                       \\
                                
                                & $\Gamma \ctvdash \ell(a_1,\ldots,a_n) : \alpha
                                \bullet \CC_1$                  
                                & $\Gamma \ctvdash a : \alpha_1 \bullet
                                \CC_2$                                  \\
                                
                                & $\CC = \CC_1 \cup \CC_2 \cup \{\alpha =_s
s_2^{\ell}, \alpha_1
                                \sub_s s_1^?\}$&        \\ 
\end{tabular}
\end{flushleft}

We are given that $(\sigma,s^g)$ is a solution for $(\Gamma(\ell :
\vsig{(s_1^?)}{s_2^{\ell}}), e, \alpha, \CC)$, that is, $\sigma$ satisfies $\CC$ and
$\sigma\alpha \sub_s s^g$. Since $(\sigma,s^g)$ satisfies $\CC_1$ and $\CC_2$,
$(\sigma, \sigma\alpha)$ and $(\sigma, \sigma\alpha_1)$ are solutions for
$(\Gamma, \ell(a_1,\ldots,a_n), \alpha, \CC_1)$ and $(\Gamma,a,\alpha_1,\CC_2)$,
respectively.
By the \ih, we have $\sigma\Gamma \vdash \sigma(\ell(a_1,\ldots,a_n)): \sigma\alpha$ and
$\sigma\Gamma \vdash \sigma a : \sigma\alpha_1$. Since $\sigma\alpha_1 \sub_s
s_1^?$, by $\Sub$ we obtain $\sigma\Gamma \vdash \sigma a : s_1^?$. Since $\sigma\alpha =
s_2^{\ell}$, by
$\TElem$ we obtain $\sigma(\Gamma(\ell : \vsig{(s_1^?)}{s_2^{\ell}})) \vdash
\sigma(\ell(a_1, \dots, a_n,
a)) : s_2^{\ell}$. By $\Sub$ we
obtain  $\sigma(\Gamma(\ell : \vsig{(s_1^?)}{s_2^{\ell}})) \vdash \sigma(\ell(a_1, \dots, a_n,
a)) : s^g$, as required.

\begin{flushleft}
\begin{tabular}{lll}
        \textit{Case} $\CTMerge$:       & $e = \ell(a_1,\ldots,a_n,a)$
                                & $\tau = \alpha$                       \\
                                
                                & $\Gamma \ctvdash \ell(a_1,\ldots,a_n) : \alpha
                                \bullet \CC_1$                  
                                & $\Gamma \ctvdash a_1 : \alpha \bullet
                                \CC_2$                                  \\
                                
                                & $\CC = \CC_1 \cup \CC_2 \cup \{\alpha =_s
s_2^{\ell}\}$&  \\ 
\end{tabular}
\end{flushleft}

We are given that $(\sigma,s^g)$ is a solution for $(\Gamma(\ell :
\vsig{(s_1^?)}{s_2^{\ell}}), e, \alpha, \CC)$, that is, $\sigma$ satisfies $\CC$ and
$\sigma\alpha \sub_s s^g$. Since $(\sigma,s^g)$ satisfies $\CC_1$ and $\CC_2$,
$(\sigma, \sigma\alpha)$ and $(\sigma, \sigma\alpha_1)$ are solutions for
$(\Gamma, \ell(a_1,\ldots,a_n), \alpha, \CC_1)$ and $(\Gamma,a,\alpha,\CC_2)$.
By the \ih, we have $\sigma\Gamma \vdash \sigma(\ell(a_1,\ldots,a_n)): \sigma\alpha$ and
$\sigma\Gamma \vdash \sigma a : \sigma\alpha_1$. Since $\sigma\alpha =
s_2^{\ell}$, by
$\TMerge$ we obtain $\sigma(\Gamma(\ell : \vsig{(s_1^?)}{s_2^{\ell}})) \vdash
\sigma(\ell(a_1, \dots, a_n,
a)) : s_2^{\ell}$. By $\Sub$ we
obtain  $\sigma(\Gamma(\ell : \vsig{(s_1^?)}{s_2^{\ell}})) \vdash \sigma(\ell(a_1, \dots, a_n,
a)) : s^g$, as required.

\begin{flushleft}
\begin{tabular}{lll}
        \textit{Case} $\CTMatch$:       & $e = a_1 ~\match_{[\tau_1]}~ a_2$
                                & $\tau = \wt$                  \\
                                
                                & $\Gamma \ctvdash a_1 : \alpha_1 \bullet \CC_1$                        
                                & $\Gamma \ctvdash a_2 : \alpha_2 \bullet \CC_2$                                        \\
                                
                                & $\CC = \CC_1 \cup \CC_2 \cup \{\alpha_1 \sub_s
\tau_1, \alpha_2 =_s \tau_1\}$&   \\ 
\end{tabular}
\end{flushleft}

We are given that $(\sigma,\wt)$ is a solution for $(\Gamma, e, \wt, \CC)$, that
is, $\sigma$ satisfies $\CC$ and
$\sigma \wt \sub_s \wt$. Since $(\sigma,\wt)$ satisfies $\CC_1$ and $\CC_2$,
$(\sigma, \sigma\alpha_1)$ and $(\sigma, \sigma\alpha_2)$ are solutions for
$(\Gamma, a_1, \alpha_1, \CC_1)$ and $(\Gamma,a_2,\alpha_2,\CC_2)$, respectively.
By the \ih, we have $\sigma\Gamma \vdash \sigma a_1: \sigma\alpha_1$ and
$\sigma\Gamma \vdash \sigma a_2 : \sigma\alpha_2$. Since $\sigma\alpha_1 \sub_s
\sigma\tau_1$, by $\Sub$ we obtain
$\sigma\Gamma \vdash \sigma a_1: \sigma\tau_1$. Since $\sigma\alpha_2 =
\sigma \tau_1$, by
$\TMatch$ we obtain $\sigma\Gamma \vdash \sigma(a_1 ~\match_{[\tau_1]}~ a_2) : \wt$, as required.
\end{proof}

\begin{defn}[Normal form of typing derivation]
A typing derivation is in {\em normal form} if it does not have successive
applications of rule \textsc{[\Sub]}.
\end{defn}

\begin{thm}[Completeness of constraint typing]
\label{thm:completeness}
Suppose that $\pi = \Gamma \ctvdash e : \tau \bullet \CC$. Write $V(\pi)$ for
the set of all type variables mentioned in the last rule used to derive $\pi$
and write $\sigma \backslash V(\pi)$ for the substitution that is undefined for all the
variables in $V(\pi)$ and otherwise behaves like $\sigma$. If $(\sigma, s^g)$ is a
solution for $(\Gamma,e)$ and $dom(\sigma) \cap V(\pi) = \varnothing$, then
there is some solution $(\sigma',s^g)$ for
$(\Gamma,e,\tau,\CC)$ such that $\sigma' \backslash V(\pi) = \sigma$.
\end{thm}
\begin{proof}
By induction on the given constraint typing derivation in normal form, but we
must take care with fresh names of variables. We just detail the most noteworthy cases of this
proof.

\begin{flushleft}
\begin{tabular}{lll}
        \textit{Case} $\CTElem$:        & $e = \ell(a_1,\ldots,a_n,a)$
                                & $\tau = \alpha$                       \\
                                
                                & $\pi_1 = \Gamma \ctvdash \ell(a_1,\ldots,a_n) : \alpha
                                \bullet \CC_1$                  
                                & $\pi_2 = \Gamma \ctvdash a : \alpha_1 \bullet
                                \CC_2$                                  \\
                                
                                & $\CC = \CC_1 \cup \CC_2 \cup \{\alpha =_s
s_2^{\ell}, \alpha_1
                                \sub_s s_1^?\}$
                                & $V(\pi) = \{\alpha, \alpha_1\}$       \\
                                
                                & $\sortof{\Gamma}{a} ~\neq s_2^{\ell}$ & \\
\end{tabular}
\end{flushleft}

From the assumption that $(\sigma,s^g)$ is a solution for $(\Gamma(\ell :
\vsig{(s_1^?)}{s_2^{\ell}}), \ell(a_1,\ldots,a_n,a))$ and $dom(\sigma) \cap V(\pi) = \varnothing$, we have
$\sigma(\Gamma(\ell : \vsig{(s_1^?)}{s_2^{\ell}})) \vdash \sigma(\ell(a_1, \dots, a_n,
a)) : s^g$. This can be derived from: 
\begin{inparaenum}[1)] 
\item $\TMerge$,
\item $\TElem$ or
\item $\Sub$. 
\end{inparaenum}
In all those cases, we must exhibit a substitution $\sigma'$ such that: (a) $\sigma'
\backslash V(\pi)$ agrees with $\sigma$; (b) $\sigma'\alpha \sub_s s^g$; (c) $\sigma'$
satisfies $\CC_1$ and $\CC_2$; and (d) $\sigma'$ satisfies $\{\alpha =_s
s_2^{\ell},
\alpha_1 \sub_s s_1^?\}$. We reason by cases as follows:  

\begin{enumerate}
\item By $\TMerge$ we assume that $s^g = s_2^{\ell}$ and we know that
$\sigma\Gamma \vdash \sigma(\ell(a_1,\ldots,a_n)):
s_2^{\ell}$ and $\sigma\Gamma \vdash \sigma a : s_2^{\ell}$. But since we cannot find a
type $s_3^{\ell}$ such that $s_3^{\ell} \sub_s s_2^{\ell}$, $\sigma\Gamma \vdash \sigma a :
s_2^{\ell}$ cannot be derived even from $\Sub$. Thus $\TMerge$ is not a relevant case.

\item By $\TElem$ we assume that $s^g = s_2^{\ell}$ and we know that
$\sigma\Gamma \vdash \sigma(\ell(a_1,\ldots,a_n)):
s_2^{\ell}$ and $\sigma\Gamma \vdash \sigma a : s_1^?$. By the \ih, there are
solutions $(\sigma_1,s_2^{\ell})$ for $(\Gamma, \ell(a_1,\ldots,a_n), \alpha, \CC_1)$
and $(\sigma_2,s_1^?)$ for $(\Gamma,a,\alpha_1,\CC_2)$, and
$dom(\sigma_1) \backslash V(\pi_1) = \varnothing = dom(\sigma_2) \backslash
V(\pi_2)$. Define $\sigma' = \{\alpha \mapsto s_2^{\ell}, \alpha_1 \mapsto s_1^?\}
\cup \sigma \cup \sigma_1 \cup \sigma_2$. Conditions (a), (b), (c) and (d) are obviously satisfied. Thus, we
see that $(\sigma',s^g)$ is a solution for $(\Gamma(\ell :
\vsig{(s_1^?)}{s_2^{\ell}}), \ell(a_1,\ldots,a_n,a), \alpha, \CC)$.

\item By $\Sub$ we assume that $s_2^{\ell} \sub_s s^g \in \Gamma$ and we know that
$\sigma(\Gamma(\ell : \vsig{(s_1^?)}{s_2^{\ell}})) \vdash \sigma(\ell(a_1, \dots, a_n, a))
: s_2^{\ell}$. This must be derived from $\TElem$, similar to case (2). 
\end{enumerate}

\begin{flushleft}
\begin{tabular}{lll}
        \textit{Case} $\CTMerge$:       & $e = \ell(a_1,\ldots,a_n,a)$
                                & $\tau = \alpha$                       \\
                                
                                & $\pi_1 = \Gamma \ctvdash \ell(a_1,\ldots,a_n) : \alpha
                                \bullet \CC_1$                  
                                & $\pi_2 = \Gamma \ctvdash a : \alpha \bullet
                                \CC_2$                                  \\
                                
                                & $\CC = \CC_1 \cup \CC_2 \cup \{\alpha =_s
s_2^{\ell}\}$
                                & $V(\pi) = \{\alpha\}$ \\
                                
                                & $\sortof{\Gamma}{a} = s_2^{\ell}$ & \\
\end{tabular}
\end{flushleft}

From the assumption that $(\sigma,s^g)$ is a solution for $(\Gamma(\ell :
\vsig{(s_1^?)}{s_2^{\ell}}), \ell(a_1,\ldots,a_n,a))$ and $dom(\sigma) \cap V(\pi) = \varnothing$, we have
$\sigma(\Gamma(\ell : \vsig{(s_1^?)}{s_2^{\ell}})) \vdash \sigma(\ell(a_1, \dots, a_n,
a)) : s^g$. This can be derived from:
\begin{inparaenum}[1)] 
\item $\TMerge$,
\item $\TElem$ or
\item $\Sub$. 
\end{inparaenum}
In all those cases, we must exhibit a substitution $\sigma'$ such that: (a) $\sigma'
\backslash V(\pi)$ agrees with $\sigma$; (b) $\sigma'\alpha \sub_s s^g$; (c) $\sigma'$
satisfies $\CC_1$ and $\CC_2$; and (d) $\sigma'$ satisfies $\{\alpha =_s
s_2^{\ell}\}$. We reason by cases as follows:  

\begin{enumerate}
\item By $\TMerge$ we assume that $s^g = s_2^{\ell}$ and we know that
$\sigma\Gamma \vdash \sigma(\ell(a_1,\ldots,a_n)):
s_2^{\ell}$ and $\sigma\Gamma \vdash \sigma a : s_2^{\ell}$. By the \ih, there are
solutions $(\sigma_1,s_2^{\ell})$ for $(\Gamma, \ell(a_1,\ldots,a_n), \alpha, \CC_1)$
and $(\sigma_2,s_2^{\ell})$ for $(\Gamma,a,\alpha,\CC_2)$, and
$dom(\sigma_1) \backslash V(\pi_1) = \varnothing = dom(\sigma_2) \backslash
V(\pi_2)$. Define $\sigma' = \{\alpha \mapsto s_2^{\ell}\}
\cup \sigma \cup \sigma_1 \cup \sigma_2$. Conditions (a), (b), (c) and (d) are obviously satisfied. Thus, we
see that $(\sigma',s^g)$ is a solution for $(\Gamma(\ell :
\vsig{(s_1^?)}{s_2^{\ell}}), \ell(a_1,\ldots,a_n,a), \alpha, \CC)$.

\item By $\TElem$ we assume that $s^g = s_2^{\ell}$ and we know that
$\sigma\Gamma \vdash \sigma(\ell(a_1,\ldots,a_n)):
s_2^{\ell}$ and $\sigma\Gamma \vdash \sigma a : s_1^?$. But, because of the
application condition of $\TElem$, we cannot find a type $s_1^{\ell}$ for
$\sigma a$ such that $s_1^{\ell} \sub_s s_1^?$, $\sigma\Gamma \vdash \sigma a :
s_1^?$ cannot be derived from $\Gen$. Likewise, since we cannot find a type
$s_3^{\ell}$ for $\sigma a$ such that $s_3^{\ell} \sub_s s_1^?$, $\sigma\Gamma \vdash \sigma a :
s_1^?$ cannot be derived even from $\Sub$. Thus $\TElem$ is not a relevant case.

\item By $\Sub$ we assume that $s_2^{\ell} \sub_s s^g \in \Gamma$ and we know that
$\sigma(\Gamma(\ell : \vsig{(s_1^?)}{s_2^{\ell}})) \vdash \sigma(\ell(a_1, \dots, a_n, a))
: s_2^{\ell}$. This must be derived from $\TMerge$, similar to case (1). 
\end{enumerate}

\begin{flushleft}
\begin{tabular}{lll}
        \textit{Case} $\CTMatch$:       & $e = a_1 ~\match_{[\tau_1]}~ a_2$
                                & $\tau = \wt$                  \\
                                
                                & $\pi_1 = \Gamma \ctvdash a_1 : \alpha_1 \bullet \CC_1$                        
                                & $\pi_2 = \Gamma \ctvdash a_2 : \alpha_2 \bullet \CC_2$                                        \\
                                
                                & $\CC = \CC_1 \cup \CC_2 \cup \{\alpha_1 \sub_s
\tau_1, \alpha_2 =_s \tau_1\}$
                                & $V(\pi) = \{\alpha_1, \alpha_2, \tau_1\}$ if $\tau_1 \in \VV$         \\
                                & & $V(\pi) = \{\alpha_1, \alpha_2\}$ if $\tau_1
\notin \VV$     \\ 
\end{tabular}
\end{flushleft}

From the assumption that $(\sigma,\wt)$ is a solution for $(\Gamma, a_1 ~\match_{[\tau_1]}~ a_2)$ and $dom(\sigma) \cap V(\pi) = \varnothing$, we have
$\sigma\Gamma \vdash \sigma(a_1 ~\match_{[\tau_1]}~ a_2) : \wt$. This must be
derived from $\TMatch$, we know that $\sigma\Gamma \vdash \sigma a_1: \sigma\tau_1$ and
$\sigma\Gamma \vdash \sigma a_2 : \sigma\tau_1$. By the \ih, there are solutions
$(\sigma_1,\sigma\tau_1)$ for $(\Gamma, a_1, \alpha_1, \CC_1)$ and
$(\sigma_2,\sigma\tau_1)$ for $(\Gamma,a_2,\alpha_2,\CC_2)$. We must exhibit a substitution $\sigma'$ such that: (a) $\sigma'
\backslash V(\pi)$ agrees with $\sigma$; (b) $\sigma' \wt \sub_s \wt$; (c) $\sigma'$
satisfies $\CC_1$ and $\CC_2$; and (d) $\sigma'$ satisfies $\{\alpha_1 \sub_s
\tau_1, \alpha_2 =_s \tau_1\}$. Define $\sigma'' = \{\alpha_1 \mapsto s^g,
\alpha_2 \mapsto s^g\} \cup \sigma \cup \sigma_1 \cup \sigma_2$, where $s^g \in \caD$. Moreover, define $\sigma' =
\sigma'' \cup \{\tau_1 \mapsto s^g\}$ if $\tau_1 \in \VV$ and $\sigma' =
\sigma''$ otherwise. Conditions (a), (b), (c) and (d) are obviously satisfied. Thus, we
see that $(\sigma',\wt)$ is a solution for $(\Gamma, (a_1
~\match_{[\tau_1]}~ a_2), \wt, \CC)$.
\end{proof}

The constraint resolution algorithm always terminates. More formally:
\begin{thm}[Termination of algorithm] ~
\label{thm:termination}
\begin{enumerate}
\item \label{thm_i1} the algorithm halts, either by failing or by returning a
substitution, for all $\CC$;
\item \label{thm_i2}if the algorithm returns $\sigma$, then $\sigma$ is a solution for
$\CC$;
\end{enumerate} 
\end{thm}

We can already sketch a proof of Theorem \ref{thm:termination} following Pierce
\cite{pierce:02}.
\begin{proof}
For part \ref{thm_i1}, define the {\em degree} of a constraint set $\CC$ to be
the pair $(m,n)$, where $m$ is the number of constraints in $\CC$ and $n$ is the number
of subtyping constraints in $\CC$. The algorithm terminates immediately (with
success in the case of an empty constraint set or failure for an equation involving two
different decorated sorts) or makes recursive calls to itself with
a constraint set of lexicographically smaller degree.

For part \ref{thm_i2}, by induction on the number of recursive calls in the
computation of the algorithm.
\end{proof}


\section{Conclusion}
\label{sec:conclusion}

In this paper we have presented a type system for the pattern matching
constructs of \tom. The system is composed of type
checking and type inference algorithms with subtyping over sorts. Since \tom
also implements associative pattern matching over variadic operators, we were 
interested in defining both a way to distinguish these from syntactic
operators and checking and inferring their types.

We have obtained the following: our type inference system is sound and complete
w.r.t. checking, showed by Theorems \ref{thm:completeness} and \ref{thm:soundness}. This
is the first step towards an effective implementation, thus leading to a safer
\tom.
However, we still need to investigate type unicity that we believe to hold under our assumptions of non-overloading and non-multiple inheritance.

As we have considered a subset of the \tom language, future work will 
focus on extending the type system to handle the other constructions of the
language such as anti-patterns \cite{KirchnerKM-2007,kopetz:08}. 
As a slightly more prospective research area, we also want parametric polymorphism over types for \tom: our type system will therefore have
to be able to handle that as well.  

\section*{Acknowledgements} We would like to acknowledge the numerous discussions we had
in the Protheo and Pareo teams, especially with Paul Brauner, on these topics during these
last years as well as the constructive and useful comments done by the anonymous
refereees.

\bibliography{typing}
\bibliographystyle{plain}

\end{document}

%% file: rule2009-kmt-macros.tex
\newcommand{\ih}{induction hypothesis\xspace}

\newcommand{\tom}{\textsf{Tom}\xspace}
\newcommand{\Java}{\textsf{Java}\xspace}
\newcommand{\java}{\Java}

\newcommand{\FF}{\ensuremath{\Sigma}\xspace}
\newcommand{\TT}{\ensuremath{\mathcal{T}}\xspace}
\newcommand{\XX}{\ensuremath{\mathcal{X}}\xspace}
\newcommand{\VV}{\ensuremath{\mathcal{V}}\xspace}

\newcommand{\FV}{\ensuremath{\FF_v}\xspace}

\newcommand{\caD}{\ensuremath{\mathcal{D}}\xspace}

\newcommand{\caS}{\ensuremath{\mathcal{S}}\xspace}

\newcommand{\caF}{\ensuremath{\mathcal{F}}\xspace}

\newcommand{\TFX}{\ensuremath{\TT(\caF,\XX)}\xspace}
\newcommand{\TFVX}{\ensuremath{\TT(\caF \cup \caF_v,\XX)}\xspace}
\newcommand{\caC}{\ensuremath{\mathcal{C}}\xspace}
\newcommand{\CC}{\ensuremath{C}\xspace}
\newcommand{\VC}{\ensuremath{\VV(\CC)}\xspace}
\newcommand{\xx}{\ensuremath{x^*}\xspace}

\newcommand{\TYDV}{\ensuremath{\TT_{ype}(\caD \cup \{wt\},\VV)}\xspace}
\newcommand{\NN}{\ensuremath{\mathbb{N}}\xspace}
\newcommand{\ZZ}{\ensuremath{\mathbb{Z}}\xspace}

\newcommand{\match}{\mathrel{\mbox{$\prec\hspace{-0.5em}\prec$}}}

\newcommand{\TVar}{{\small \textsc{T-Var}}}
\newcommand{\TSVar}{{\small \textsc{T-SVar}}}
\newcommand{\TFun}{{\small \textsc{T-Fun}}}
\newcommand{\TEmpty}{{\small \textsc{T-Empty}}}
\newcommand{\TElem}{{\small \textsc{T-Elem}}}
\newcommand{\TMerge}{{\small \textsc{T-Merge}}}
\newcommand{\Sub}{{\small \textsc{Sub}}}
\newcommand{\Gen}{{\small \textsc{Gen}}}
\newcommand{\TMatch}{{\small \textsc{T-Match}}}
\newcommand{\TConj}{{\small \textsc{T-Conj}}}

\newcommand{\TRule}{{\small \textsc{T-Rule}}}

\newcommand{\CTVar}{{\small \textsc{CT-Var}}}
\newcommand{\CTSVar}{{\small \textsc{CT-SVar}}}
\newcommand{\CTFun}{{\small \textsc{CT-Fun}}}
\newcommand{\CTEmpty}{{\small \textsc{CT-Empty}}}
\newcommand{\CTElem}{{\small \textsc{CT-Elem}}}
\newcommand{\CTStar}{{\small \textsc{CT-Star}}}
\newcommand{\CTMerge}{{\small \textsc{CT-Merge}}}
\newcommand{\CTMatch}{{\small \textsc{CT-Match}}}
\newcommand{\CTConj}{{\small \textsc{CT-Conj}}}

\newcommand{\CTRule}{{\small \textsc{CT-Rule}}}

\newcommand{\sub}{\mathrel{\mbox{{\footnotesize $<:$}}}}
\newcommand{\transsub}{\mathrel{\mbox{$\sub\hspace{-1em}\footnotesize{^*}~~$}}}
\newcommand{\wt}[0]{{\ensuremath{wt}}}
\newcommand{\ctvdash}{\ensuremath{\vdash_{ct}}}

\newcommand{\sortof}[2]{\texttt{sortOf}\ensuremath{(#1,#2)}}

\newcommand{\sig}[2]{\ensuremath{{#1}\rightarrow{#2}}}
\newcommand{\vsig}[2]{\ensuremath{{#1}^*\rightarrow{#2}}}

\newcommand{\ie}{\textit{i.e.}}

\def\vs{\vspace{-1.5mm}}